\documentclass[english,runningheads,a4paper]{llncs}
\usepackage[T1]{fontenc}
\usepackage[latin9]{inputenc}
\usepackage{color}
\usepackage{float}
\usepackage{amsmath}
\usepackage{graphicx}
\usepackage{amssymb}
\makeatletter

\floatstyle{ruled}
\newfloat{algorithm}{tbp}{loa}
\providecommand{\algorithmname}{Algorithm}
\floatname{algorithm}{\protect\algorithmname}

\newtheorem{thm}{\protect\theoremname}
\newtheorem{defn}[thm]{\protect\definitionname}
\newtheorem{cor}[thm]{\protect\corollaryname}
\newtheorem{prop}[thm]{\protect\propositionname}
\newtheorem{lem}[thm]{\protect\lemmaname}

\makeatother

\usepackage{babel}
\providecommand{\corollaryname}{Corollary}
\providecommand{\definitionname}{Definition}
\providecommand{\lemmaname}{Lemma}
\providecommand{\propositionname}{Proposition}
\providecommand{\remarkname}{Remark}
\providecommand{\theoremname}{Theorem}

\begin{document}

\title{Improved Approximation Algorithms for Computing $k$ Disjoint Paths
Subject to Two Constraints }

\titlerunning{Improved Approximation Algorithms for $k$BCP}

\author{\textcolor{black}{Longkun Guo$^{1}$ \thanks{\textcolor{black}{This project was supported by the Natural Science
Foundation of Fujian Province (2012J05115), Doctoral Fund of Ministry
of Education of China for Young Scholars (20123514120013) and Fuzhou
University Development Fund (2012-XQ-26). Longkun Guo (lkguo@fzu.edu.cn) is the corresponding author.}
}, Hong Shen$^{2,3}$, Kewen Liao$^{3}$}}
\authorrunning{Longkun Guo, Hong Shen and Kewen Liao}
\institute{
\textcolor{black}{{} $^{1}$College of Mathematics and Computer Science,
Fuzhou University, China}\\\textcolor{black}{{} $^{2}$School of Information Science and Technology, Sun Yat-Sen University, China}\\
\textcolor{black}{{} $^{3}$School of Computer Science, University of
Adelaide, Australia}}

\maketitle
\begin{abstract}
For a given graph $G$ with positive integral cost and delay on edges,
distinct vertices $s$ and $t$, cost bound $C\in Z^{+}$ and delay
bound $D\in Z^{+}$, the $k$ bi-constraint path ($k$BCP) problem
is to compute $k$ disjoint $st$-paths subject to $C$ and $D$.
This problem is known NP-hard, even when $k=1$ \cite{garey1979computers}.
This paper first gives a simple approximation algorithm with factor-$(2,2)$,
i.e. the algorithm computes a solution with delay and cost bounded
by $2*D$ and $2*C$ respectively. Later, a novel improved approximation
algorithm with ratio $(1+\beta,\,\max\{2,\,1+\ln\frac{1}{\beta}\})$
is developed by constructing interesting auxiliary graphs and employing
the cycle cancellation method. As a consequence, we can obtain a factor-$(1.369,\,2)$
approximation algorithm by setting $1+\ln\frac{1}{\beta}=2$ and a
factor-$(1.567,\,1.567)$ algorithm by setting $1+\beta=1+\ln\frac{1}{\beta}$.
Besides, by setting $\beta=0$, an
approximation algorithm with ratio $(1,\, O(\ln n))$, i.e. an algorithm
with only a single factor ratio $O(\ln n)$ on cost, can be immediately
obtained. To the best of our knowledge, this is the first non-trivial
approximation algorithm for the $k$BCP problem that strictly obeys
the delay constraint. \end{abstract}

\section{Introduction}

In real networks, there are many applications that require quality
of service and some degree of robustness simultaneously. Typically,
the quality of service (QoS) related problem requires routing between
the source node and the destination node to satisfy several constraints
simultaneously, such as bandwidth, delay, cost and energy consumption.
Nevertheless, in networks, some time-critical applications also require
routing to remain functioning while edge or vertex failure occurs.
A common solution is to compute $k$ disjoint paths that satisfy the
QoS constraints, and use one path as an \emph{active} path whilst
the other paths as \emph{backup} paths. The routing traffic is carried
on the active path, and switched to the disjoint backup paths while
an edge or vertex failure occurs on the active path. However, for
some time-critical applications even the time to discover failures
of routing and restore data transmission in backup paths is too long
for them. For such applications, packages are routed via $k$ paths
simultaneously, and the traffic is switched from failed paths to functioning
paths if edge or vertex failures occur, such that routing can tolerate
$k-1$ edge (vertex) failures. Therefore, given cost and delay as
the QoS constraints, the \emph{disjoint QoS Path problem }arises as
below:
\begin{defn}
For a graph $G=(V,\, E)$ and a pair of distinct vertices $s,\, t\in V$,
a cost function $c:E\rightarrow Z^{+}$, a delay function $d:E\rightarrow Z^{+}$,
a cost bound $C\in Z^{+}$ and a delay bound $D\in Z^{+}$, the\emph{
$k$-disjoint QoS Paths problem} is to compute $k$ disjoint $st$-paths
$P_{1},\dots,P_{k}$, such that $\underset{i=1,\dots,k}{\sum}c(P_{i})\leq C$
and $d(P_{i})\leq D$ for every $i=1,\dots,k$.
\end{defn}
This problem is NP-hard even when all edges of $G$ are with cost
0 \textcolor{black}{\cite{li1989cft}}, which results in the difficulty
to approximate the\emph{ }$k$-disjoint QoS Paths problem. An alternative
method is to compute $k$ disjoint with total cost bounded by $C$
and delay bounded by $D$ ( equal to $kD$ in Definition 1), and then
route the packages via the paths according to their urgency priority,
i.e., route urgent packages via paths of low delay whilst deferrable
ones via paths of high delay of the $k$ disjoint paths. Therefore,
The disjoint bi-constraint path problem arises as in the following:
\begin{defn}
(The $k$ disjoint bi-constraint path problem, $k$BCP) For a graph
$G=(V,E)$ with a pair of distinct vertices $s,t\in V$, a cost function
$c:E\rightarrow R^{+}$, a delay function $d:E\rightarrow R^{+}$,
a cost bound $C\in Z^{+}$ and a delay bound $D\in R^{+}$, the\emph{
$k$-disjoint bi-constraint path problem} is to calculate $k$ disjoint
$st$-paths $P_{1},\dots,P_{k}$, such that $\underset{i=1,\dots,k}{\sum}c(P_{i})\leq C$
and $\underset{i=1,\dots,k}{\sum}d(P_{i})\leq D$.
\end{defn}
This paper will focus on bifactor approximation algorithms for the
$k$BCP problem, which are introduced as below:
\begin{defn}
An algorithm $A$ is a bifactor $\left(\alpha,\,\beta\right)$-approximation
for the $k$BCP problem, if and only if for every instance of $k$BCP\emph{,
$A$} computes $k$ disjoint $st$-paths of which the delay sum and
the cost sum are bounded by $\alpha*D$ and $\beta*C$ respectively.
\end{defn}

Since a $\beta$-approximation with the single factor ratio on cost
is identical to a bifactor $\left(1,\,\beta\right)$-approximation,
we use them interchangeably in the text.

\subsection{Related work}

This $k$BCP problem is NP-hard even when $k=1$ \cite{garey1979computers}.
To the best of our knowledge, this paper is the first one that presents
non-trivial approximation algorithms for the $k$BCP problem formally.
However, a number of papers have addressed problems closely related
to $k$BCP, in particular the $k$ restricted shortest path problem
($k$RSP), which is to calculate $k$ disjoint $st$-paths of minimum
cost-sum under the delay constraint $\underset{i=1,\dots,k}{\sum}d(P_{i})\leq D$.
An algorithm with bifactor approximation ratio $(2,\,2)$ has been
developed in \cite{guopdcat} for general $k$, while no approximation
solution that strictly obeys the delay (or cost) constraint is known
even when $k=2$. For a positive real number $r$, bifactor ratio of $(1+\frac{1}{r},\, r(1+\frac{2(\log r+1)}{r})(1+\epsilon))$
and $(1+\frac{1}{r},\, r(1+\frac{2(\log r+1)}{r}))$ have been achieved
respectively in \cite{orda2004efficient,chao2007new} for the case
$k=2$ and under the assumption that the delay of each path in the
optimal solution of $k$RSP is bounded by $\frac{D}{k}$.

Special cases of this problem have been studied. When the delay constraint
is removed, this problem is reduced to the min-sum problem, which
is to calculate $k$ disjoint paths with the total cost minimized.
This problem is known polynomially solvable \cite{suurballe1974dpn}.
Moreover, when $k=1$, the problem reduces to the single bi-constraint
path (BCP) problem, which is known as the basic QoS routing problem
\cite{garey1979computers} and admits full polynomial time approximation
scheme (FPTAS) \cite{garey1979computers,lorenz2001simple}. Recently,
the single BCP problem is still attracting considerable interests of the
researchers. The strongest result known is a ($1+\epsilon$)-approximation
due to Xue et al \cite{xue2008polynomial}.

\textcolor{black}{Additionally, }when the cost constraint is removed,\textcolor{black}{{}
the disjoint QoS problem reduces to the length bounded disjoint path
problem of finding two disjoint paths with the length of each path
constrained by a given bound. This problem is a variant of the min-Max
problem of finding two disjoint paths with the length of the longer
path minimized . Both of the two problems are known to be }\textcolor{black}{\emph{NP-}}\textcolor{black}{complete
\cite{li1989cft}, and with the best possible approximation ratio
of 2 in digraphs \cite{li1989cft}, which can be achieved by applying
the algorithm for the min-sum problem in \cite{suurballe1974dpn,suurballe1984qmf}.
Contrastingly, the min-min problem of finding two paths with the length
of the shorter path minimized is NP-complete and doesn't admit $K$
approximation for any $K\geq1$ \cite{polyvertexMinMin,xu2006caa,bhatia2006finding}.
The problem remains NP-complete and admits no polynomial time approximation
scheme in planar digraphs \cite{guo2012tcs}. }

\subsection{Our techniques and results}

The main result of this paper is a factor-$(1+\beta,\,\max\{2,\,1+\ln\frac{1}{\beta}\})$
approximation algorithm for any $0<\beta\leq1$ for the $k$BCP problem.
The main idea of the algorithm is firstly to compute $k$-disjoint
paths with delay-sum bounded by $\alpha D$ and cost-sum bounded by
$(2-\alpha)*C$, where $0\leq\alpha\leq2$ is a real number, and secondly
to improve the computed $k$ paths by novelly combining cycle cancellation
\cite{orda2004efficient} and cost-bounded auxiliary graph construction
\cite{xue2008polynomial}. The key technique to prove the algorithm's
approximation ratio is using definite integral to compute a close
form for the sum of the cost increment during the improving phase.

As a consequence of the main result, we can obtain a factor-$(1.369,\,2)$
approximation algorithm by setting $1+\ln\frac{1}{\beta}=2$, and
a factor-$(1.567,\,1.567)$ algorithm by setting $1+\beta=1+\ln\frac{1}{\beta}$
and slightly modifying our algorithm (to improve either cost or delay
that is with worse ratio). Nevertheless, by slightly modifying our
ratio proof, we show that an approximation algorithm with ratio $(1,\, O(\ln n))$,
i.e. an algorithm with single factor ratio of $O(\ln n)$ on cost,
can be immediately obtained by setting $\beta=0$. To the best of
our knowledge, this is the first non-trivial approximation algorithm
for the $k$BCP problem that strictly obeys the delay constraint.

We note that our algorithms are with pseudo-polynomial time complexity,
since the auxiliary graph we construct is of size $O(C*n)$. However,
by using the classic polynomial time approximation scheme design technique
\cite{garey1979computers}, i.e. for any small $\epsilon>0$ setting
the cost of every edge to $\left\lfloor \frac{c(e)}{\frac{\epsilon C}{n}}\right\rfloor $
in $G$ before the construction of auxiliary graph, we can immediately
obtain a polynomial time algorithm with ratio $((1+\beta)*(1+\epsilon),\,\max\{2,\,1+\ln\frac{1}{\beta}\}*(1+\epsilon))$.
We shall omit the details due to the paper length limitation.

\section{An improved approximation algorithm for computing $k$ disjoint bi-constraint
paths}

This section will first present a simple approximation method for
computing $k$-disjoint paths with delay-sum bounded by $\alpha D$
and cost-sum bounded by $(2-\alpha)*C$, where $0\leq\alpha\leq2$
is a real number, and secondly improve the computed $k$ paths by
balancing the value of $\alpha$ and $2-\alpha$. Though the presented
simple algorithm is with worse ratio than that of the algorithm for
$k=2$ in \cite{orda2004efficient}, it suits the improving phase
better.

\subsection{A basic approximation algorithm}

Observing that the difficulty of computing $k$-disjoint bi-constraint
paths mainly comes from the two given constraints, the key idea of
our algorithm is to deal with one new constraint $B$ instead of the
two given constraints $C$ and $D$. Our algorithm firstly assigns
a new mixed cost $b(e)=\frac{c(e)}{C}+\frac{d(e)}{D}$ to every edge
in graph, and secondly computes $k$ disjoint paths with the new cost
sum bounded by $B=\frac{C}{C}+\frac{D}{D}=2$. Note that the second
step can be accomplished in polynomial time by employing the SPP algorithm
due to Suurballe and Tarjan \cite{suurballe1974dpn,suurballe1984qmf}.
The detailed algorithm is as in Algorithm \ref{alg:basic-appro}.

\begin{algorithm}
\textbf{Input:} A graph $G=(V,E)$, each edge $e$ with cost $c(e)$
and delay $d(e)$, a given cost constraint $C\in Z^{+}$ and delay
constraint $D\in Z^{+}$;

\textbf{Output:} $k$ disjoint paths $P_{1},\, P_{2}\,\dots,P_{k}$.
\begin{enumerate}
\item Set the new cost of edge $e$ as $b(e)=\frac{c(e)}{C}+\frac{d(e)}{D}$;
\item Compute the $k$ disjoint paths $P_{1},\, P_{2}\,\dots,P_{k}$ in
$G$ by using Suurballe and Tarjan's algorithm \cite{suurballe1974dpn,suurballe1984qmf},
such that $\sum_{i=1}^{k}\sum_{e\in P_{i}}b(e)$ is minimized;
\item Return $P_{1},\, P_{2}\,\dots,P_{k}$.
\end{enumerate}
\caption{\label{alg:basic-appro}A basic approximation algorithm for the $k$-BCP
problem}
\end{algorithm}

The time complexity and performance guarantee of Algorithm \ref{alg:basic-appro}
is given by the following theorem:
\begin{thm}
\label{thm:baseAlgo}Algorithm \ref{alg:basic-appro} runs in $O(km\log_{1+\frac{m}{n}}n)$
time, and computes $k$-disjoint paths with delay-sum bounded by $\alpha D$
and cost-sum bounded by $(2-\alpha)*C$ , where $0\leq\alpha\leq2$
is a real number.\end{thm}
\begin{proof}
The main part of Algorithm \ref{alg:basic-appro} takes $O(km\log_{1+\frac{m}{n}}n)$
to compute $k$-disjoint paths by using Surrballe and Tarjan's algorithm
\cite{suurballe1974dpn,suurballe1984qmf}, and other parts of the
algorithm take trivial time. Hence the time complexity of the algorithm
is $O(km\log_{1+\frac{m}{n}}n)$.

It remains to show the approximation ratio. To make the proof concise,
we denote by $OPT$ an optimal solution for the $k$-disjoint BCP
paths problem, and $SOL$ the solution of Algorithm \ref{alg:basic-appro}.
Obviously $\sum_{e\in OPT}b(e)\leq2$ holds. Then since the $k$ disjoint
paths is with minimum new cost, we have
\begin{equation}
\sum_{e\in SOL}b(e)\leq\sum_{e\in OPT}b(e)\leq2.\label{eq:b(e) bound}
\end{equation}
 Assume the delay-sum of the algorithm is $\alpha$ times of $d(OPT)$,
then following Algorithm \ref{alg:basic-appro} $0\leq\alpha\leq2$ holds. Therefore, we
have $\sum_{e\in SOL}b(e)=\sum_{i=1}^{k}\sum_{e\in P_{i}}b(e)=\alpha+\frac{c(SOL)}{c(OPT)}$.
From Inequality (\ref{eq:b(e) bound}), $\alpha+\frac{c(SOL)}{c(OPT)}\leq2$
holds. That is, $c(SOL)\leq(2-\alpha)c(OPT)\leq(2-\alpha)C$. This
completes the proof.
\end{proof}
Note that $\alpha$ differs for different instances, i.e. Algorithm
\ref{alg:basic-appro} may return a solution with cost $2*c(OPT)$
and delay $0$ for some instances, while a solution with cost 0 and
delay $2*d(OPT)$ for other instances. Hence, the bifactor approximation
ratio for Algorithm \ref{alg:basic-appro} is actually $(2,2)$.

In real networks, the two given constraints may not be of equal importance,
say, delay is far more important comparing to cost. In this case,
applications require that the delay of the resulting solution is bounded
by $(1+\beta)D$, where $0<\beta<1$ is a positive real number. Apparently,
we could get an algorithm similar to Algorithm \ref{alg:basic-appro}
excepting setting the new cost as $b(e)=\beta\frac{c(e)}{C}+\frac{d(e)}{D}$.
The ratio of the new algorithm is given as below:
\begin{cor}
\label{cor:By-setting-the}By setting the new cost as $b(e)=\beta\frac{c(e)}{C}+\frac{d(e)}{D}$
for a given real number $0<\beta<1$ , Algorithm \ref{alg:basic-appro}
returns $k$ paths with delay-sum bounded by $\alpha D$ and cost-sum
bounded by $\frac{1+\beta-\alpha}{\beta}*C$ , where $0\leq\alpha\leq1+\beta$
is a real number. Therefore the ratio of the algorithm is $(1+\beta,\,1+\frac{1}{\beta})$.
\end{cor}
The proof of Corollary \ref{cor:By-setting-the} is omitted here,
since it is very similar to the proof of Theorem \ref{alg:basic-appro}.
According to Corollary \ref{cor:By-setting-the}, our algorithm can
bound the delay-sum of the $k$-disjoint path by $(1+\beta)D$ for
any $0<\beta<1$, by relaxing the cost constraint to $(1+\frac{1}{\beta})*C$.
For example, if $\beta=0.01$, then the bifactor approximation ratio
of the algorithm is $(1.01,\,101)$. Thus, the algorithm decrease
the delay of the $k$-disjoint paths at a high price. In the next
subsection, we shall develop an improved method that pays less to
make delay-sum of the $k$-disjoint paths bounded by $(1+\beta)D$.

\subsection{The improving phase}

To make the delay of the solution resulting from Algorithm \ref{alg:basic-appro}
bounded by $(1+\beta)$D, our improving phase is, basically a greedy
method, using the so-called cycle cancellation to improve the disjoint
paths in iterations until a solution with the best possible ratio
$(1+\beta,\,\max\{2,1+\ln\frac{1}{\beta}\})$ is obtained. The cycle
cancellation method is an approach of using cycles to change the edges
of the disjoint paths, which first appears in \cite{orda2004efficient}
and is derived from the following proposition that can be immediately
obtained from flow theory \cite{ahuja1993network}:
\begin{prop}
\label{pro:cycle}Let $P_{1},\, P_{2}\,\dots,\, P_{k}$ and $Q_{1},\, Q_{2}\,\dots,Q_{k}$
be two sets of $k$ disjoint $st$-paths in $G$, $\overline{G}$
be $G$ excepting that all edges of $P_{1},\, P_{2}\,\dots,P_{k}$
are reversed, and $O$ be a cycle in $\overline{G}$. Then
\begin{enumerate}
\item The edges of $P_{1},\, P_{2}\,\dots,\, P_{k}$ and $O$, excepting
the pairs of parallel edges with opposite direction, compose $k$-disjoint
paths;
\item There exist a set of edge disjoint cycles $O_{1},\,\dots,\, O_{h}$
in $\overline{G}$, such that the edges of $P_{1},\, P_{2}\,\dots,\, P_{k}$
and $O_{1},\,\dots,\, O_{h}$, excepting the pairs of parallel edges
with opposite direction, compose $Q_{1},\, Q_{2}\,\dots,Q_{k}$.
\end{enumerate}
\end{prop}
From the proposition above, it is obvious that there exists a set
of cycles $O_{1},\,\dots,\, O_{h}$ that can improve $k$ disjoint
QoS paths $P_{1},\, P_{2}\,\dots,P_{k}$ to an optimal solution. However,
it is hard to identify all the cycles $O_{1},\,\dots,\, O_{h}$, so
we employ a greedy approach to compute a set of cycles to obtain an
approximation approach. The improving phase is composed by iterations,
each of which computes a cycle and then uses it to improve $P_{1},\, P_{2}\,\dots,\, P_{k}$.
More precisely, to obtain a good ratio, the algorithm computes in
iteration $j$ a cycle $O_{j}$ with $\frac{d(O_{j})}{c(O_{j})}$
minimized among the cycles in $\overline{G}$. The layout of the algorithm
is as given in Algorithm \ref{alg:improved-algo}.

\begin{algorithm}
\textbf{Input:} A graph $G=(V,E)$, each edge $e$ with cost $c(e)$
and delay $d(e)$, a given cost constraint $C\in Z^{+}$ and delay
constraint $D\in Z^{+}$, disjoint QoS paths $P_{1},\, P_{2}\,\dots,P_{k}$
computed by Algorithm \ref{alg:basic-appro};

\textbf{Output:} Improved disjoint QoS paths $Q_{1},\, Q_{2}\,\dots,Q_{k}$.
\begin{enumerate}
\item \textbf{If} $\sum_{i=1}^{k}d(P_{i})\leq(1+\beta)D$ ;

\textbf{then} return $P_{1},\, P_{2}\,\dots,P_{k}$ as $Q_{1},\, Q_{2}\,\dots,Q_{k}$
, terminate;

\item Reverse direction of the edges of $P_{1},\, P_{2}\,\dots,P_{k}$ in
$G$ , set their cost to a small positive real number $0<\epsilon<\frac{1}{mnD}$,
and negative their delay;

\item Compute cycle $O_{j}$ with $c(O_{j})\leq C$, $d(O_{j})<0$ and $\frac{d(O_{j})}{c(O_{j})}$
attaining minimum, by the method given in next section;

/{*} Following clause 2 of Proposition \ref{pro:cycle}, if $\sum_{i=1}^{k}d(P_{i})\geq d(OPT)$
and $\sum_{i=1}^{k}c(P_{i})\geq0$, there always exist cycle $O_{j}$
with $c(O_{j})\leq C$ and $d(O_{j})<0$. {*}/

\item Improve $P_{1},\, P_{2}\,\dots,P_{k}$ by adding the edges of $O_{j}$
and removing the pairs of parallel edges in opposite direction;
\item Go to Step 1.
\end{enumerate}
\caption{\label{alg:improved-algo}An improved algorithm based on cycle cancellation. }
\end{algorithm}

Following clause 1 of Proposition \ref{pro:cycle}, Algorithm \ref{alg:improved-algo}
will correctly return $k$ disjoint paths. It remains to show the
cost and delay of the $k$ disjoint paths is constrained as below:
\begin{thm}
\label{thm:The-ratio-ofHD}The approximation ratio of Algorithm 2
is $(1+\beta,\:\max\{2,\,1+\ln\frac{1}{\beta}\})$. \end{thm}
\begin{proof}
For the case that $\sum_{i=1}^{k}d(P_{i})\leq(1+\beta)D$ holds before
the improving phase, the approximation ratio of Algorithm \ref{alg:improved-algo}
is obviously $(1+\beta,\,2)$.

It remains to show the ratio of the algorithm is $(1+\beta,\:1+\ln\frac{1}{\beta})$
for the case that $\sum_{i=1}^{k}d(P_{i})>(1+\beta)D$. Assume that
Algorithm \ref{alg:improved-algo} runs in $h$ iterations, the key
idea of the proof is to sum up the cost increment while using the
cycle to improve the $k$ disjoint paths in iterations, and show that
the cost sum is bounded (by giving the cost sum a close form).

Note that in the case, we have $\alpha D\geq\sum_{i=1}^{k}d(P_{i})>(1+\beta)D$,
so $\alpha>1+\beta$ holds. Let $\Delta D=d(OPT)-d(SOL)\geq(1-\alpha)d(OPT)$
and $\Delta C=c(OPT)-c(SOL)\leq(\alpha-1)c(OPT)$. Clearly, $\Delta D<0$
and $\Delta C>0$ hold. Let the cycle computed in the $j$th iteration
be $O_{j}$, then since $\frac{d(O_{j})}{c(O_{j})}$ attains minimum
in Step 3 of Algorithm \ref{alg:improved-algo}, we have $\frac{d(O_{j})}{c(O_{j})}\leq\frac{\Delta D-\sum_{i=1}^{j-1}d(O_{i})}{C}.$
That is,
\[
c(O_{j})\leq\frac{d(O_{j})}{\Delta D-\sum_{i=1}^{j-1}d(O_{i})}C.
\]
 By summing up $c(O_{j})$ in $h-1$ iterations (excluding the last
iteration), we have:

\[
\sum_{j=1}^{h-1}c(O_{j})\leq C\sum_{j=1}^{h-1}\frac{d(O_{j})}{\Delta D-\sum_{i=1}^{j-1}d(O_{i})}.
\]

Following the definition of Definite Integral, we have:

\begin{equation}
\sum_{j=1}^{h-1}\frac{d(O_{j})}{\Delta D-\sum_{i=1}^{j-1}d(O_{i})}=\sum_{j=1}^{h-1}\frac{1}{\Delta D-\sum_{i=1}^{j-1}d(O_{i})}d(O_{j})\leq\int_{\Delta D}^{\Delta D-\sum_{i=1}^{h-1}d(O_{i})}\frac{1}{x}dx,\label{eq:sum1}
\end{equation}

where the maximum is attained when $d(O_{j})=-1$ for every $j$.

Algorithm \ref{alg:improved-algo} terminates when $d(SOL)+\sum_{i=1}^{h}d(O_{i})\leq(1+\beta)D$,
so in the $h-1$ iterations $d(SOL)+\sum_{i=1}^{h-1}d(O_{i})>(1+\beta)D$
holds. That is $d(SOL)-D+\sum_{i=1}^{h-1}d(O_{i})>\beta D$, and hence
$-\Delta D+\sum_{i=1}^{h-1}d(O_{i})>\beta D>0$ holds. So we obtain
a close form for the cost sum of the $h-1$ iterations:

\begin{equation}
\int_{\Delta D}^{\Delta D-\sum_{i=1}^{h-1}d(O_{i})}\frac{1}{x}dx=\int_{-\Delta D+\sum_{i=1}^{h-1}d(O_{i}))}^{-\Delta D}\frac{1}{x}dx\leq\int_{\beta D}^{-\Delta D}\frac{1}{x}dx=\ln\frac{|\Delta D|}{\beta D}=\ln\frac{\alpha-1}{\beta}.\label{eq:sum}
\end{equation}

At last, the cost increment in the $h$th iteration is bounded by
$c(OPT)$. So the final cost is $c(SOL_{2})\leq(2-\alpha)C+C\ln\frac{\alpha-1}{\beta}+C=C(3-\alpha+\ln\frac{\alpha-1}{\beta})$,
where $SOL_{2}$ is the solution resulting from Algorithm \ref{alg:improved-algo}.

Let $f(\alpha)=3-\alpha+\ln\frac{\alpha-1}{\beta}$. Remind that $\alpha\leq2$,
so $f'(\alpha)=\frac{1}{\alpha-1}-1>0$, $f(\alpha)$ is monotonous
increasing on $\alpha$, and attains maximum while $\alpha=2$. So
we have $c(SOL_{2})\leq(1+\ln\frac{1}{\beta})c(OPT)$.

Therefore, the cost of the output of Algorithm \ref{alg:improved-algo}
is bounded by $(1+\ln\frac{1}{\beta})c(OPT)$, and delay bounded by
$(1+\beta)d(OPT)$. This completes the proof.
\end{proof}
From Theorem \ref{thm:The-ratio-ofHD}, by setting $1+\ln\frac{1}{\beta}=2$,
we can immediately obtain an improved algorithm with best possible
delay ratio under the same cost bound $2C$. That is:
\begin{cor}
By setting $1+\ln\frac{1}{\beta}=2$, we have $\beta=\frac{1}{e}$,
and hence Algorithm \ref{alg:improved-algo} is now with a bifactor
approximation ratio of $(1+\frac{1}{e},\,2)=(1.369,\,2)$.
\end{cor}
For those applications in which delay and cost are of equal importance,
by setting $1+\ln\frac{1}{\beta}=1+\beta$ and slightly modifying
Algorithm \ref{alg:improved-algo} to improve either cost or delay
that is of worse ratio, we can obtain an improved algorithm with ratio
as in the following corollary:
\begin{cor}
If $1+\ln\frac{1}{\beta}=1+\beta$, Algorithm \ref{alg:improved-algo}
is with a bifactor approximation ratio of $(1.567,\,1.567)$.
\end{cor}
Now we consider the case that $\beta=0$, i.e. the delay constraint
is strictly satisfied. In this case, Inequality (\ref{eq:sum}) in
the proof of Theorem \ref{thm:The-ratio-ofHD} will become $\sum_{j=1}^{h-1}\frac{d(O_{j})}{\Delta D-\sum_{i=1}^{j-1}d(O_{i})}\leq\int_{|\beta D=0|}^{|\Delta D|}\frac{1}{x}dx=\ln|\Delta D|\leq\ln D.$
So we have:
\begin{cor}
\label{cor:strict}When $\beta=0$, Algorithm \ref{alg:improved-algo}
is with a ratio of $(1,\, O(\ln n))$.
\end{cor}
From Corollary \ref{cor:strict}, we can see that the price of obeying
one constraint strictly is very high, i.e. it requires extra $O(\ln n)$
times of cost. However, this is the first algorithm with logarithmic
factor approximation ratio for the $k$-BCP problem with strict delay
constraint.

\section{Computing Cycle $O_{j}$ with minimum $\frac{d(O_{j})}{c(O_{j})}$}

Let $\overline{G}=(V,E)$ be $G$, excepting that the edges of $P_{1},\, P_{2}\,\dots,P_{k}$
are with direction reversed, cost sat to 0, and delay negatived. This
section will show how to compute a cycle $O$ with cost bounded by
$C$ and $\frac{d(O)}{c(O)}$ minimized in $\overline{G}$. The key
idea is firstly to construct an auxiliary graphs $H(v)$ for each
$v$ where every cycle is with cost at most $C$, secondly to compute
the cycle $O'$ with minimum $\frac{d(O')}{c(O')}$ among all cycles
in all $H(v)$s for each $v\in\overline{G}$, and thirdly to obtain
cycle $O$ with minimum $\frac{d(O(v))}{c(O(v))}$ in $\overline{G}$
according to $O'$.

\subsection{Construction of auxiliary graph $H(v)$}

The algorithm of constructing the auxiliary  graph $H(v)$ is inspired
by the method of computing a single path subject to multiple constraints
\cite{xue2008polynomial}. The full layout of the algorithm is as
shown in Algorithm \ref{alg:Construction-of-auxiliary} (An example
of such construction is as depicted in Figure \ref{fig:Construction-of-acyclic}).

\begin{algorithm}
\textbf{Input}: Graph $\overline{G}=(V,E)$, two distinct vertices
$s,\, t\in V$, a cost $c:\, e\rightarrow Z_{0}^{+}$ and a delay
$d:\, e\rightarrow Z_{0}^{+}$ on every edge $e\in E$, a cost constraint
$C$ and a delay constraint $D$;

\textbf{Output}: Auxiliary  graph $H(v)$.
\begin{enumerate}
\item For every vertex $v_{l}$ of $V$, add to $H(v)$ vertices $v_{l}^{1},\dots,v_{l}^{C}$
;
\item For every edge $e=\left\langle v_{j},v_{l}\right\rangle \in E$, add
to $H(v)$ the edges $\left\langle v_{j}^{1},\, v_{l}^{c(e)+1}\right\rangle ,\,\dots,\left\langle v_{j}^{C-c(e)},\, v_{l}^{C}\right\rangle $,
each of which is with cost $c(e)$ and delay $d(e)$;

/{*}Note that $d(e)$ can be negative in $\overline{G}=(V,E)$.{*}/

\item For all $i=2,\,\dots,\, C$, add to $H(v)$ backward edge $\left\langle v^{i},\, v^{1}\right\rangle $
with delay 0 and cost 0, where a backward edge is an edge $\left\langle v^{i},\, v^{j}\right\rangle $
where $i>j$.

/{*}$H(v)$ contains backward edges, and hence cycles, only after
adding the edges of Step 3.{*}/

\end{enumerate}
\caption{\label{alg:Construction-of-auxiliary}Construction of auxiliary graph
$H$.}
\end{algorithm}

Following Algorithm \ref{alg:Construction-of-auxiliary}, every backward
edge in the constructed auxiliary graph $H(v)$ must contain vertex
$v^{1}$. Hence every cycle in $H(v)$ contains at most one backward
edge. On the other hand, following Algorithm \ref{alg:Construction-of-auxiliary}
a cycle in $H(v)$ contains at least one backward edge. Therefore,
there exist exactly one backward edge in any cycle of $H(v)$. Because
$H(v)\setminus\{\left\langle v^{2},\, v^{1}\right\rangle ,\dots,\left\langle v^{C},\, v^{1}\right\rangle \}$
is an acyclic graph where any path is with cost at most $C$, we have:
\begin{lem}
\label{lem:cyclecost}Any cycle in $H(v)$ is with cost at most $C$.
\end{lem}
\begin{figure}
\begin{centering}
\includegraphics[width=0.7\textwidth]{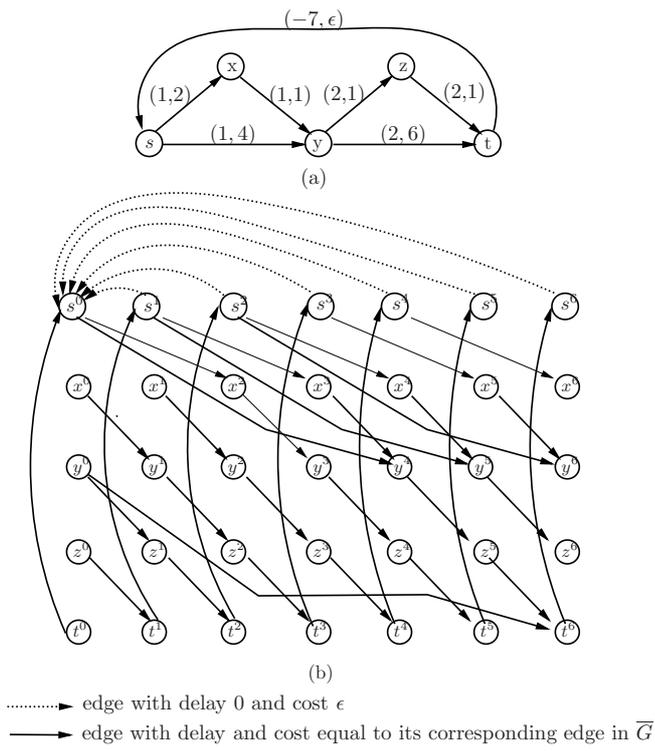}
\par\end{centering}

\caption{\label{fig:Construction-of-acyclic}Construction of auxiliary  graph
$H(v=s)$ with cost constraint $C=6$: (a) graph $\overline{G}$;
(b) auxiliary graph $H(v=s)$. The cycle $O=syts$ in $\overline{G}$
is exclude in the auxiliary graph $H(s)$ as shown in (b), keeping
the cost of $k$ disjoint paths constrained by $C=6$.}
\end{figure}

Let $O(v)$ be a cycle in $H(v)$, then following the construction
of $H(v)$, $O(v)$ apparently corresponds to a set of cycles in $\overline{G}$.
Conversely, every cycle containing $v$ in $\overline{G}$ corresponds
to a cycle in $H(v)$. Based on the observation, the following lemma
gives the key idea of computing a cycle $O$ of $\overline{G}$ with
$\frac{d(O)}{c(O)}$ minimized and cost bounded by $C$:
\begin{lem}
\label{lem:cycleinH} Let $O(v_{i})$ be a cycle with minimum $\frac{d(O(v_{i}))}{c(O(v_{i}))}$
in $H(v_{i})$, and $O(v)$ be the cycle with minimum $\frac{d(O(v))}{c(O(v))}$
among the $n$ cycles $O(v_{1}),\dots,O(v_{n})$. Assume $O$ is a
cycle with minimum $\frac{d(O)}{c(O)}$ in the set of cycles in $\overline{G}$
that correspond to $O(v)$. Then for any cycle $O'$ in $\overline{G}$
with $c(O')\leq C$, $\frac{d(O)}{c(O)}\leq\frac{d(O')}{c(O')}$ holds. \end{lem}
\begin{proof}
Suppose this lemma is not true, then there must exist in $\overline{G}$
a cycle, say $O'$, such that $\frac{d(O)}{c(O)}>\frac{d(O')}{c(O')}$
and $c(O')\leq C$ hold. Then the cycle $O'(v)$ in $H(v)$ that corresponds
to $O'$ is also with $\frac{d(O'(v))}{c(O'(v))}=\frac{d(O')}{c(O')}<\frac{d(O)}{c(O)}\leq\frac{d(O(v))}{c(O(v))}$,
contradicting with the minimality of $O(v)$ in $H(v)$. This completes
the proof.
\end{proof}

\subsection{Computing the cycle $O$ with minimum $\frac{d(O)}{c(O)}$}

The main idea of the algorithm to compute a cycle $O$ with $\frac{d(O)}{c(O)}$
minimized in $\overline{G}$ is to compute the cycle $O'$ with minimum
$\frac{d(O')}{c(O')}$ among all cycles in all $H(v)$s for each $v\in\overline{G}$.
Following Lemma \ref{lem:cycleinH}, the cycle $O$ in $\overline{G}$
is the cycle with minimum $\frac{d(O)}{c(O)}$ among the cycles in
$\overline{G}$ corresponding to all the computed $O'$s. The detailed
steps are as below:
\begin{enumerate}
\item For $i=1\, to\, n$

\begin{enumerate}
\item Construct $H(v_{i})$ for $v_{i}\in\overline{G}$ by Algorithm \ref{alg:Construction-of-auxiliary};
\item Compute cycle $O(v_{i})$ with minimum $\frac{d(O(v_{i}))}{c(O(v_{i}))}$
in $H(v_{i})$ by employing the minimum cost-to-time ratio cycle algorithm
in \cite{ahuja1993network};
\item Select $O(v)$ with minimum $\frac{d(O(v))}{c(O(v))}$ from the $n$
computed cycles $O(v_{1}),\dots,O(v_{n})$;
\end{enumerate}
\item Select the cycle $O$ with minimum $\frac{d(O)}{c(O)}$ among the
cycles in $\overline{G}$ that correspond to $O(v)$.
\end{enumerate}
Clearly, the cycle $O$ attains minimum $\frac{d(O)}{c(O)}$ in $\overline{G}$.
Besides, following Lemma \ref{lem:cyclecost} we have $c(O)\leq C$.
 Therefore the cycle $O$ is correctly the promised cycle. This completes
the proof of the approximation ratio.

\section{Conclusion}

This paper gave a novel approximation algorithm with ratio $(1+\beta,\,\max\{2,\,1+\ln\frac{1}{\beta}\})$
for the $k$BCP problem based on improving a simple $(\alpha,\,2-\alpha)$-approximation
algorithm by constructing interesting auxiliary graphs and employing
the cycle cancellation method. By setting $\beta=0$, an approximation
algorithm with bifactor ratio $(1,\,O(\ln n))$,
i.e. an $O(\ln n)$-approximation algorithm can be obtained immediately.
To the best of our knowledge, it is the first non-trivial approximation
algorithm for this problem that obeys the delay constraint strictly.
We are now investigating whether any constant factor approximation
algorithm exists for computing a solution that strictly obey the delay
constraint.

\bibliographystyle{plain}
\bibliography{disjointQoS}

\end{document}